\documentclass{article}

\usepackage{booktabs} 
\usepackage[ruled]{algorithm2e} 

\SetAlFnt{\small}
\SetAlCapFnt{\small}
\SetAlCapNameFnt{\small}
\SetAlCapHSkip{0pt}
\IncMargin{-\parindent}

\usepackage{tabularx}

\usepackage[margin = 1in]{geometry}
\usepackage{bbding}
\usepackage[nointegrals]{wasysym}
\usepackage[utf8]{inputenc}
\usepackage{xspace}
\usepackage{mathtools}
\usepackage{url}
\usepackage{subcaption}
\usepackage{tikz}
\usetikzlibrary{arrows, fit}
\usetikzlibrary{backgrounds,automata,calc}
\usetikzlibrary{decorations.pathreplacing}
\usepackage{csquotes}
\usepackage{amsmath,amsfonts,amsthm}
\usepackage{pgfplots}

\DeclareMathOperator*{\argmax}{arg\,max}

\usepackage{enumerate}
\usepackage{paralist}
\usepackage[shortlabels]{enumitem}
\usepackage{appendix}
\usepackage[capitalize]{cleveref}
\usepackage{thmtools}
\usepackage{mathtools}
\usepackage{natbib}
\usepackage{changepage}

\usepackage{commath}

\newcommand{\CMCA}{\textsc{ComputeMinCostAllocation}\xspace}

\newcommand{\MAXUSW}{\texttt{MAX-USW}\xspace}

\newcommand{\R}{\mathbb{R}}
\newcommand{\Z}{\mathbb{Z}}

\newcommand{\lex}{\texttt{lex}}

\newtheorem{theorem}{Theorem}[section]
\newtheorem{lemma}[theorem]{Lemma}

\newtheorem{prop}[theorem]{Proposition}

\newtheorem*{theorem*}{Theorem}

\theoremstyle{definition}

\theoremstyle{definition}

\title{Weighted Notions of Fairness with Binary Supermodular Chores}
\author{Vignesh Viswanathan and Yair Zick \\
\texttt{\{vviswanathan, yzick\} @ umass.edu}}
\date{}

\begin{document}

\maketitle

\begin{abstract}
We study the problem of allocating indivisible chores among agents with binary supermodular cost functions. In other words, each chore has a marginal cost of $0$ or $1$ and chores exhibit increasing marginal costs (or decreasing marginal utilities). In this note, we combine the techniques of \citet{viswanathan2022generalyankee} and \citet{barman2023chores} to present a general framework for fair allocation with this class of valuation functions. Our framework allows us to generalize the results of \citet{barman2023chores} and efficiently compute allocations which satisfy weighted notions of fairness like weighted leximin or min weighted $p$-mean malfare for any $p \ge 1$. 
 \end{abstract}

\section{Introduction}\label{sec:intro}
Consider the problem of dividing a set of tasks among the employees of a firm. Employees can specify which tasks they would be happy to complete; they may also have upper limits for the number of tasks they would be happy to complete. However, all tasks must be assigned to some employee and completed by them. In such cases, it is best both for the employees and the firm that the tasks get assigned in a way that is {\em efficient} --- we maximize the number of tasks employees will be happy completing and {\em fair} --- we distribute tasks evenly among the employees. In addition, firm hierarchies must be taken into consideration as well; different employees may have different weights (or entitlements).

This problem can be naturally modelled as a fair allocation of indivisible chores problem. If we assume agents are only allowed to submit `yes' or `no' preferences for each task and tasks have increasing costs to scale, they are said to have binary supermodular cost functions. This specific problem has been extensively studied in the unweighted setting by \citet{barman2023chores}. In particular, \citet{barman2023chores} show that leximin allocations, envy free up to one chore (EF1) allocations and maxmin fair allocations can be computed efficiently. 

In this work, we generalize these results to the weighted setting where each agent has a potentially different weight (or entitlement). We show that weighted leximin allocations as well as min weighted $p$-mean malfare allocations can be computed efficiently. More generally, we present sufficient conditions for the efficient computation of any {\em justice criterion} (similar to the results in \citet{viswanathan2022generalyankee}).
\section{Preliminaries}\label{sec:prelims}
We use $[k]$ to denote the set $\{1, 2, \dots, k\}$. Given a set $S$ and an element $o$, we use $S + o$ and $S - o$ to denote the sets $S \cup \{o\}$ and $S \setminus \{o\}$ respectively. 

We have a set of $n$ {\em agents} $N = [n]$ and set of $m$ {\em indivisible chores} $O = \{o_1, o_2, \dots, o_m\}$. Each agent $i \in N$ has a {\em cost function} $c_i: 2^O \rightarrow \R_{\ge 0}$; $c_i(S)$ denotes the cost of the set of items $S$ according to agent $i$. 
For each agent $i\in N$, we let $\Delta_i(S,o) = c_i(S+o) - c_i(S)$ denote the marginal cost of the chore $o$ to the bundle $S$ for the agent $i$.
Each agent $i \in N$ has an associated weight $w_i$ which corresponds to their entitlement.

We assume agents have binary supermodular cost functions. That is, for every agent $i \in N$, $c_i$ satisfies the following conditions:
\begin{inparaenum}[(a)]
    \item $c_i(\emptyset) = 0$, 
    \item for any $S \subseteq O$, and $o \in O \setminus S$, $\Delta_i(S, o) \in \{0, 1\}$, and 
    \item for any $S \subseteq T \subseteq O$ and $o \in O \setminus T$, $\Delta_i(S, o) \le \Delta_i(T, o)$.
\end{inparaenum}

An {\em allocation} $X = (X_0, X_1, \dots, X_n)$ is an $(n+1)$-partition of the set of chores $O$. $X_i$ denotes the set of chores allocated to agent $i$ and $X_0$ denotes the set of unallocated chores. Our goal is to compute {\em complete} fair allocations --- allocations where $X_0 = \emptyset$. Given an allocation $X$, we refer to $-c_i(X_i)$ as the utility of agent $i$ under the allocation $X$. We also define the utility vector of an allocation $X$ as the vector $\vec u^X = (-c_1(X_1), -c_2(X_2), \dots, -c_n(X_n))$. 

These definitions can be extended to the weighted context as well. We define the weighted utility of an agent $i \in N$ under the allocation $X$ as $\frac{-c_i(X_i)}{w_i}$. Similarly, we define the weighted utility vector of an allocation $X$ as the vector $\vec e^X = \big (\frac{-c_1(X_1)}{w_1}, \frac{-c_2(X_2)}{w_2}), \dots, \frac{-c_n(X_n)}{w_n} \big )$.

A vector $\vec y \in \R^n$ {\em lexicographically dominates} a vector $\vec z \in \R^n$ (denoted $\vec y \succ_{\lex} \vec z$) if there exists a $k \in [n]$ such that for all $j \in [k-1]$, $y_j = z_j$ and $y_k > z_k$. For readability, we use $X \succ_{\lex} Y$ to denote $\vec u^X \succ_{\lex} \vec u^Y$. If $X \succ_{\lex} Y$, we sometimes say $X$ lexicographically dominates $Y$.

\subsection{Justice Criteria}\label{sec:justice-criteria}
There are several {\em justice criteria} (or {\em fairness metrics}) in the literature. We focus on a few specific ones in this paper.

\noindent\textbf{Utilitarian Social Welfare: } The utilitarian social welfare (USW) of an allocation $X$ is given by $\sum_{i \in N} - c_i(X_i)$ --- the sum of the utilities of all the agents. 
An allocation is referred to as \MAXUSW if it maximizes the utilitarian social welfare.

\noindent\textbf{Weighted Leximin:} An allocation is said to be weighted leximin if it maximizes the {\em weighted} utility of the agent with the least weighted utility and subject to that, maximizes the {\em weighted} utility of the agent with the second least weighted utility, and so on. This can be formalized using the sorted weighted utility vector. The sorted weighted utility vector of an allocation $X$ (denoted by $\vec {se}^X$) corresponds to its weighted utility vector sorted in ascending order. An allocation $X$ is weighted leximin if for any allocation $Y$, $\vec {se}^X \succeq_{\lex} \vec {se}^Y$.

\noindent\textbf{Min Weighted $p$-mean malfare:} For any $p\ge 1$, an allocation $X$ is min weighted $p$-mean malfare if it minimizes $\sum_{i \in N}w_i c_i(X_i)^p$. $p$-mean malfare functions have been used extensively in fair machine learning \citep{cousins2021axiomatic} and are natural justice criteria for fair chore allocation.

More generally, a justice criterion $\Psi$ is a way to compare utility vectors of allocations. We use $\vec u^{X} \succ_{\Psi} \vec u^Y$ to say $X$ is better than $Y$ according to justice criterion $\Psi$. For example, if $\Psi$ corresponds to the utilitarian social welfare metric, we say $\vec u^X \succeq_{\Psi} \vec u^Y$ if $X$ has a utilitarian social welfare greater than or equal to $Y$, i.e. $\sum_{i \in N}c_i(X_i) \le \sum_{i \in N}c_i(Y_i)$.  We usually replace $\vec u^X \succeq_{\Psi} \vec u^Y$ with $X \succeq_{\Psi} Y$ for readability.
We assume that $\Psi$ induces a total ordering over the set of all allocations. 
Given a justice criterion $\Psi$, our goal is to a compute a $\Psi$-maximal complete allocation $X$ --- more formally, we would like to compute a complete allocation $X$ such that for all complete allocations $Y$, we have $X \succeq_{\Psi} Y$.

\subsection{Understanding Binary Supermodular Costs}
We establish a few important properties of the cost function $c_i$. 
The first result (from \citet{barman2023chores}) shows that we can efficiently compute a maximum size partial allocation of chores with a utilitarian social welfare of $0$. More formally,

\begin{theorem}[\citet{barman2023chores}]\label{thm:min-cost-partial-allocation}
When agents have binary supermodular cost functions, there exists an efficient algorithm (we call \CMCA) which can compute an allocation $X$ with USW $0$ such that for no other allocation $Y$ with a USW of $0$, we have $|\bigcup_{i \in N} X_i| < |\bigcup_{i \in N} Y_i|$.
\end{theorem}

Our next result shows that any allocation $X$ can be decomposed into two allocations $X^1$ and $X^0$ such that for every agent $i \in N$, $c_i(X_i) = |X^1_i|$ and $c_i(X^0_i) = 0$. This result is similar to the decomposition result in \citet{cousins2023bivalued}. 
Consistent with the terminology from \citet{cousins2023bivalued}, we refer to $X^0$ as the clean allocation and $X^1$ as the supplementary allocation.

\begin{lemma}\label{lem:decomposition}
When agents have binary supermodular cost functions, any allocation $X$ can be decomposed into two allocations $X^0$ and $X^1$ such that for any agent $i \in N$, 
\begin{inparaenum}[(a)]
    \item $X^1_i \cup X^0_i = X_i$, 
    \item $X^1_i \cap X^0_i = \emptyset$,
    \item $c_i(X^0_i) = 0$, and
    \item $c_i(X_i) = |X^1_i|$.
\end{inparaenum}
\end{lemma}
\begin{proof}
For each agent $i \in N$, let $X^0_i$ be the largest subset of $X_i$ such that $c_i(X^0_i) = 0$; let $X^1_i$ contain all the chores in $X_i \setminus X^0_i$. This construction trivially satisfies conditions (a), (b) and (c). We show (d) next.

Fix an agent $i \in N$. Let $X^1_i = \{o_1, \dots, o_k\}$ and let $Z^j = \{o_1, \dots, o_j\}$ (with $Z^{0} = \emptyset$). Consider the telescoping sum
\begin{align*}
    c_i(X_i) = c_i(X^0_i) + \sum_{j \in [k]} \Delta_i(X^0_i \cup Z^{j-1}, o_j)
\end{align*}

Since $c_i(X^0_i) = 0$ and $c_i$ has binary marginal gains, the above inequality gives us, $c_i(X_i) \le |X^1_i|$. If the inequality is strict, for some chore $o_j \in X^1_i$, we must have $\Delta_i(X^0_i \cup Z^{j-1}, o_j) = 0$. From the supermodularity of $c_i$, it must be the case that $\Delta_i(X^0_i, o_j) = 0$ as well. This means that $c_i(X^0_i + o_j) = 0$ which contradicts our construction of $X^0_i$. Therefore, equality must hold and we must have $c_i(X_i) = |X^1_i|$. This means (d) holds as well for each agent $i \in N$ and our proof is complete.
\end{proof}

Decompositions need not be unique; that is, there may be several clean and supplementary allocations which satisfy (a)--(d) above. We use the notation $X = X^0 \cup X^1$ to denote a decomposition of $X$ into (an arbitrary) $X^0$ and $X^1$ that satisfies the properties (a)--(d) above. More generally, for any two allocations $X$ and $Y$, we use $X \cup Y$ to denote the allocation where every agent $i \in N$ receives the bundle $X_i \cup Y_i$.

\section{General Yankee Swap for Chores}
We are now ready to present our general framework for chore allocation. Our framework is simple, we first compute a partial allocation with a USW of $0$ according to Theorem \ref{thm:min-cost-partial-allocation}. We then allocate the remaining chores one by one greedily according to $\Psi$. 

Similar to \citet{viswanathan2022generalyankee}, this greedy allocation is done via a gain function $\phi$. $\phi$ takes as input the utility vector of an allocation $X$ and an agent $i$ and outputs a real-valued $b$-dimensional vector corresponding to the `value' of allocating a chore with marginal cost $1$ to agent $i$ according to the justice criterion $\Psi$. 
The higher the value of the gain function $\phi(\vec u^X, i)$, the better it is to allocate a chore (with marginal cost $1$) to agent $i$. When $b > 1$, $\phi(\vec u^X, i) > \phi(\vec u^X, j)$ if $\phi(\vec u^X, i) \succ_{\lex} \phi(\vec u^X, j)$. For readability, we sometimes use $\phi(X, i)$ to denote $\phi(\vec u^X, i)$.

We present pseudocode in Algorithm \ref{algo:general-yankee-chores}. We first compute a partial allocation using \CMCA (Theorem \ref{thm:min-cost-partial-allocation}) and iteratively update this allocation using the gain function. 
More specifically, we compute gain function values for each agent with respect to the current partial allocation and then choose the agent with the highest gain function value; we break ties by choosing the agent with the higher index. 
We then give this agent an arbitrary unallocated chore. 
This proceeds until there are no unallocated chores left. 
To clearly differentiate the chores with marginal cost $0$ computed initially and the chores allocated greedily with marginal cost $1$, we store them separately using $X^0$ and $X^1$. 
$X^0$ consists of the initially constructed partial allocation and $X^1$ consists of the chores allocated greedily in the second phase. 
This does not come with any abuse of notation; as we shall see, $X^0$ and $X^1$ correspond to a valid decomposition of the complete allocation given by their union $X^0 \cup X^1$.

\subsection{Sufficient Conditions for the Correctness of \Cref{algo:general-yankee-chores}}
The conditions for General Yankee Swap for Chores are similar to those for General Yankee Swap \citep{viswanathan2022generalyankee} and are given below. 
These conditions are defined for arbitrary vectors, but it is useful to think of $\vec x$, $\vec y$ and $\vec z$ as utility vectors.

\begin{description}
    \item[(C1) --- Pareto Dominance:] For any two vectors $\vec x, \vec y \in \Z^n_{\le 0}$, if $x_h \ge y_h$ for all $h \in N$, then $\vec x \succeq_{\Psi} \vec y$. Equality holds iff $\vec x = \vec y$.
    \item[(C2) --- Gain Function:] $\Psi$ admits a gain function $\phi$ that maps each possible utility vector and agent to a $b$-dimensional real-valued vector satisfying:
    \begin{description}
        \item[(G1)] For any vector $\vec x \in \Z^n_{\le 0}$ and any $i, j \in [n]$, Let $\vec y \in \Z^n_{\le 0}$ be the vector that results from starting at $\vec x$ and reducing $x_i$ by $1$. 
        Similarly, let $\vec z \in \Z^n_{\le 0}$ be the vector resulting reducing $x_j$ by $1$. 
        If $\phi(\vec x, i) \ge \phi(\vec x, j)$, $\vec y \succeq_{\Psi} \vec z$. 
        Equality holds iff $\phi(\vec x, i) = \phi(\vec x, j)$.
        \item[(G2)] For any two vectors $\vec x, \vec y \in \Z^n_{\le 0}$ and $i \in [n]$, if $x_i \ge y_i$, then $\phi(\vec x, i) \ge \phi(\vec y, i)$; equality holds if $x_i = y_i$.
    \end{description}
\end{description}

\begin{algorithm}[t]
    \caption{General Yankee Swap for Chores}
    \label{algo:general-yankee-chores}
    \DontPrintSemicolon
    $X^0 = (X^0_0, X^0_1, \dots, X^0_n) \gets$ \CMCA (Theorem \ref{thm:min-cost-partial-allocation})\;
    $X^1 = (X^1_0, X^1_1, \dots, X^1_n) \gets (G, \emptyset, \emptyset, \dots, \emptyset)$\;
    \While{$|X^0_0| > \sum_{j \in N} |X^1_i|$}{
        $S \gets \argmax_{j \in N} \phi(X^0 \cup X^1, j)$\;
        $i \gets \max_{j \in S} j$\;
        $X^1_i \gets X^1_i + o$ for some chore $o \in X^0_0 \setminus (\bigcup_{j \in N} X^1_j)$\;
        $X^1_0 \gets X^1_0 - o$\;
    }
    \Return $X^0 \cup X^1$\;
\end{algorithm}

\subsection{Analysis}
We first show that $X^0$ and $X^1$ are a valid decomposition of $X^0 \cup X^1$.
\begin{lemma}
At any iteration of Algorithm \ref{algo:general-yankee-chores}, $X^0$ and $X^1$ is a valid decomposition of $X^0 \cup X^1$. In other words, for any agent $i \in N$, the following holds:
\begin{inparaenum}[(a)] 
    \item $X^1_i \cap X^0_i = \emptyset$,
    \item $c_i(X^0_i) = 0$, and
    \item $c_i(X^1_i \cup X^0_i) = |X^1_i|$.
\end{inparaenum}
\end{lemma}
\begin{proof}
(a) and (b) hold trivially; we show (c) via an argument similar to Lemma \ref{lem:decomposition}. 

Fix an agent $i \in N$. 
Let $X^1_i = \{o_1, \dots, o_k\}$ and let $Z^j = \{o_1, \dots, o_j\})$ (with $Z^{0} = \emptyset$). 
Consider the telescoping sum
\begin{align*}
    c_i(X_i) = c_i(X^0_i) + \sum_{j \in [k]} \Delta_i(X^0_i \cup Z^{j-1}, o_j)
\end{align*}
Since $c_i(X^0_i) = 0$ and $c_i$ has binary marginal gains, we have that $c_i(X_i) \le |X^1_i|$. 
If the inequality is strict, for some chore $g_j \in X^1_i$, we have $\Delta_i(X^0_i \cup Z^{j-1}, o_j) = 0$. 
Since $c_i$ is supermodular, $\Delta_i(X^0_i, o_j) = 0$ as well, and therefore $c_i(X^0_i + o_j) = 0$.
Since $o_j \in X^0_0$, we can construct a larger USW $0$ partial allocation by moving $o_j$ to $X^0_i$ --- contradicting Theorem \ref{thm:min-cost-partial-allocation}. 
\end{proof}
We use this result to show the correctness of the algorithm.

\begin{theorem}
For any justice criterion $\Psi$ that satisfies (C1) and (C2) with gain function $\phi$, Algorithm \ref{algo:general-yankee-chores} run with input gain function $\phi$ efficiently computes a $\Psi$ maximizing allocation.
\end{theorem}
\begin{proof}
The computational efficiency of Algorithm \ref{algo:general-yankee-chores} is trivial given Theorem \ref{thm:min-cost-partial-allocation}, so we do not explicitly prove it. We only show correctness. 

Let $X^0$ and $X^1$ be the allocations output by Algorithm \ref{algo:general-yankee-chores}. We use $X$ to denote $X^0 \cup X^1$.
Let $Y = Y^0 \cup Y^1$ be a non-redundant allocation that maximizes $\Psi$.
If there are multiple such $Y$, pick one that lexicographically dominates all others. 

If for all $h \in N$, $|X^1_h| \le |Y^1_h|$, then $X$ maximizes $\Psi$ (since $\Psi$ respects Pareto dominance according to C1) --- we are done. Assume for contradiction that this does not hold.

This means that there is some agent $i$ whose utility under $X$ is strictly lower than under $Y$, i.e. $|X^1_i| > |Y^1_i|$. 
Let $i \in N$ be the agent with highest $\phi(X, i)$ in $X$ such that $|X^1_i| > |Y^1_i|$; if there are multiple agents we break ties in favor of the highest index agent.

If there is no such agent $j$ such that $|X^1_j| < |Y^1_j|$, then it must be the case that $\sum_{h \in N}|X^0_h| < \sum_{h \in N} |Y^0_h|$ contradicting Theorem \ref{thm:min-cost-partial-allocation}. Let $j \in N$ be an agent in $X$ such that $|X^1_j| < |Y^1_j|$.

We construct an allocation $Z$ starting at $Y^0$ and $Y^1$ and moving an arbitrary chore from $Y^1_j$ to $Z^1_i$. Note that if $\phi(Y, i) \ge \phi(Z, j)$, then $Z \succeq_{\Psi} Y$ with equality holding if and only if $\phi(Y, i) = \phi(Z, j)$ (using G1).

Let $W = W^0 \cup W^1$ be the non-redundant allocation maintained by Algorithm \ref{algo:general-yankee-chores} at the start of the iteration where $i$ received its final chore. By our choice of iteration, we must have 
\begin{align}
    \phi(W, i) \ge \phi(W, j). \text{ If equality holds, then } i > j \label{eq:wi-greater-wj}
\end{align}

We can use this to compare $Y$ and $Z$ using (G2):
\begin{align}
    \phi(Y, i) \ge \phi(W, i) \ge \phi(W, j) \ge \phi(Z, j) \label{eq:yi-greater-zj}
\end{align}

If any of these weak inequalities are strict, we are done since $\phi(Y, i) > \phi(Z, j)$. If all the weak inequalities are equalities, then $\phi(Y, i) = \phi(Z, j)$ which in turn implies $Y =_{\Psi} Z$. However, since $i > j$ (from \eqref{eq:wi-greater-wj}), $Z$ lexicographically dominates $Y$ which contradicts our assumption on $Y$. 

Therefore, $X$ must be a $\Psi$ maximizing allocation.
\end{proof}

We also show that for any gain function $\phi$, Algorithm \ref{algo:general-yankee-chores} computes a \MAXUSW allocation.

\begin{prop}
For any gain function $\phi$, Algorithm \ref{algo:general-yankee-chores} computes a complete \MAXUSW allocation.
\end{prop}
\begin{proof}
Assume for contradiction that there exists a complete allocation $Y = Y^0 \cup Y^1$ with a higher USW than $X^0 \cup X^1$. Since $c_i(Y_i) = |Y^1_i|$ for each agent $i \in N$, we must have $\sum_{h \in N} |Y^1_h| < \sum_{h \in N} |X^1_h|$. Since both allocations are complete and do not have any unallocated goods, this in turn implies  $\sum_{h \in N} |Y^0_h| > \sum_{h \in N} |X^0_h|$ which contradicts Theorem \ref{thm:min-cost-partial-allocation}.
\end{proof}

\section{Applying General Yankee Swap for Chores}
In this section, we apply Algorithm \ref{algo:general-yankee-chores} to the justice criteria described in Section \ref{sec:justice-criteria}. We start with the weighted leximin allocation.

\begin{theorem}
Algorithm \ref{algo:general-yankee-chores} run with input gain function $\phi(X, i) = (\frac{-c_i(X_i)}{w_i}, w_i)$ computes a weighted leximin allocation.   
\end{theorem}
\begin{proof}
Formally, for any two allocations $X \succ_{\Psi} Y$ iff $\vec {se}^X$ lexicographically dominates $\vec {se}^Y$. It is easy to see that $\Psi$ trivially satisfies (C1) and $\phi$ trivially satisfies (G2); so we only show (G1).

Assume $\phi(\vec x, i) > \phi(\vec x, j)$ for some agents $i, j \in N$ and vector $\vec x \in \Z^n_{\le 0}$. Let $\vec y$ be the allocation that results from starting at $\vec x$ adding $-1$ to $x_i$ and let $\vec z$ be the allocation that results from starting at $\vec x$ adding $-1$ to $x_j$.

If $\phi(\vec x, i) > \phi(\vec x, j)$, then one of the following two cases must be true.

\textbf{Case 1:} $\frac{x_i}{w_i} > \frac{x_j}{w_j}$. Since it is always better to subtract utility from agents with higher weighted utility, we have $\vec y \succ_{\Psi} \vec z$.

\textbf{Case 2:} $\frac{x_i}{w_i}= \frac{x_j}{w_j}$ and $w_i > w_j$. If this is true, we have $\frac{y_j}{w_j} = \frac{z_i}{w_i}$ by assumption. However, $\frac{y_i}{w_i} = \frac{x_i - 1}{w_i} > \frac{x_j - 1}{w_j} = \frac{z_j}{w_j}$. Since the two vectors differ only in the values of the indices $j$ and $i$, we can conclude that $\vec y \succ_{\Psi} \vec z$.

This implies that when $\phi(\vec x, i) > \phi(\vec x, j)$, we have $\vec y \succ_{\Psi} \vec z$ as required. 

When $\phi(\vec x, i) = \phi(\vec x, j)$, we must have $\frac{x_i}{w_j}= \frac{x_j}{w_j}$ and $w_i = w_j$. This gives us $\frac{y_j}{w_j} = \frac{z_i}{w_i}$ and $\frac{y_i}{w_i} = \frac{x_i - 1}{w_i} = \frac{x_j - 1}{w_j} = \frac{z_j}{w_j}$ which implies that $\vec y =_{\Psi} \vec z$. 
\end{proof}
We can similarly show that min weighted $p$-mean malfare allocations can be computed efficiently.

\begin{theorem}
Algorithm \ref{algo:general-yankee-chores} run with input gain function $\phi(X, i) = w_i[c_i(X_i)^p - (c_i(X_i)+1)^p]$ computes a min weighted $p$-mean welfare allocation for any $p \ge 1$.   
\end{theorem}
\begin{proof}
Formally, $X \succ_{\Psi} Y$  if and only if $\sum_{h \in N} w_h c_h(X_h)^p < \sum_{h \in N} w_h c_h(Y_h)^p$.
It is easy to see that $\Psi$ satisfies (C1) and $\phi$ satisfies (G2).  Similar to the previous result, we only show (G1). For any vector $\vec x \in \Z^n_{\le 0}$, consider two agents $i$ and $j$. Let $\vec y$ be the vector that results from starting at $\vec x$ and adding $-1$ to $x_i$ and $\vec z$ be the vector that results from starting at $\vec x$ and adding $-1$ to $x_j$. We have: 

\begin{align*}
    \vec y \succeq_{\Psi} \vec z
    &\Leftrightarrow  \sum_{h \in N} w_h(-y_h)^p  \le \sum_{h \in N} w_h (-z_h)^p \\
    &\Leftrightarrow \sum_{h \in N} w_h (-y_h)^p - \sum_{h \in N} w_h (-x_h)^p \le \sum_{h \in N} w_h (-z_h)^p - \sum_{h \in N} w_h (-x_h)^p \\
    &\Leftrightarrow w_i [(-x_i - 1)^p - (-x_i)^p] \le w_j [(-x_j - 1)^p - (-x_j)^p] \\
    &\Leftrightarrow \phi(\vec x, i) \ge \phi(\vec x, j)
\end{align*}

We can replace the inequalities above with equalities and show that $$\vec y =_{\Psi} \vec z \Leftrightarrow \phi(\vec x, i) = \phi(\vec x, j)$$.
\end{proof}

\bibliographystyle{plainnat}
\bibliography{abb,literature}

\end{document}